%% file: arc-cp-tour.tex
\documentclass[a4paper,11pt]{article}

\usepackage{amsmath, amssymb, amsthm,complexity}
\usepackage{todonotes}
\usepackage[ruled,vlined,linesnumbered]{algorithm2e}
\usepackage{algorithmic}
\usepackage{comment}

\usepackage{thmtools}
\usepackage{thm-restate}
\usepackage{hyperref}

\usepackage{cleveref}
\usepackage{tikz}
\usepackage{amsmath}
\usetikzlibrary{shapes}
\usetikzlibrary{plotmarks}

\usetikzlibrary{arrows,calc,shapes,decorations.pathreplacing}

\usepackage{wrapfig}
\usepackage{lscape}
\usepackage{rotating}
\usepackage{epstopdf}

\usepackage{microtype}

\usepackage{authblk}

\def\margin{2.9cm}
\usepackage[left=\margin,right=\margin,top=\margin,bottom=\margin]{geometry}

\title{The Parameterized Complexity of Packing Arc-Disjoint Cycles in Tournaments}

\date{}

\author[1]{R. Krithika}
\author[1]{Abhishek Sahu}
\author[1,2]{Saket Saurabh}
\author[3]{Meirav Zehavi}
\affil[1]{The Institute of Mathematical Sciences, HBNI, Chennai, India. \\\texttt{\{rkrithika|asahu|saket\}@imsc.res.in}}

\affil[2]{University of Bergen, Bergen, Norway}
\affil[3]{Ben-Gurion University, Beersheba, Israel. \texttt{meiravze@bgu.ac.il}}

\usepackage{complexity} 
\usepackage{enumerate}
\usepackage{amsmath, amssymb, amsthm,complexity}
\usepackage{euler}
\usepackage{authblk}
\usepackage{todonotes}

\def\margin{2.9cm}
\usepackage[left=\margin,right=\margin,top=\margin,bottom=\margin]{geometry}

\usepackage{thmtools}
\usepackage{thm-restate}
\usepackage{hyperref}
\usepackage{cleveref}

\usepackage{authblk}

\usepackage{todonotes}
\usepackage{microtype}
\usepackage{multirow}
\usepackage{amsmath,amssymb}
\usepackage{comment}
\usepackage{enumerate}

\usepackage{xspace}

\usepackage{listings}
\usepackage{color}
\usepackage{cite}
\usepackage{complexity}

\usepackage{graphicx}
\RequirePackage{fancyhdr}
 \usepackage{xcolor}
\usepackage{boxedminipage}

\theoremstyle{plain}
\newtheorem{theorem}{Theorem}

\newtheorem{lemma}[theorem]{Lemma}

\theoremstyle{definition}
\newtheorem{observation}[theorem]{Observation}
\newtheorem{proposition}[theorem]{Proposition}

\newtheorem{reduction rule}{Reduction Rule}[section]
\newtheorem{definition}{Definition}[section]

\newcommand{\cpt}{\textsc{Arc-Disjoint Cycle Packing in Tournaments}}
\newcommand{\cp}{\textsc{Arc-Disjoint Cycle Packing}}
\newcommand{\vcp}{\textsc{Vertex-Disjoint Cycle Packing}}

\newcommand{\oo}{\mathcal{O}}
\newcommand{\ostar}{\mathcal{O}^\star}
\newcommand{\mc}{\mathcal{C}}
\newcommand{\mf}{\mathcal{F}}

\newcommand{\mi}{\mathcal{I}}
\newcommand{\mj}{\mathcal{J}}

\DeclareMathOperator{\fas}{fas}
\DeclareMathOperator{\first}{first}
\DeclareMathOperator{\second}{mid}
\DeclareMathOperator{\head}{head}
\DeclareMathOperator{\tail}{tail}

\makeatletter
\newcommand{\injective}[2][]{\ext@arrow 0359\rightarrowfill@{#1}{#2}}
\makeatother

\newcommand{\defparprob}[4]{
  \vspace{2mm}
\noindent\fbox{
  \begin{minipage}{0.96\textwidth}
  \begin{tabular*}{\textwidth}{@{\extracolsep{\fill}}lr} \textsc{#1}  & {\bf{Parameter:}} #3
\\ \end{tabular*}
  {\bf{Input:}} #2  \\
  {\bf{Question:}} #4
  \end{minipage}
  }
  \vspace{2mm}
}

\begin{document}

\maketitle

\begin{abstract}
\input{abstract.tex}

\end{abstract}

\section{Introduction}
\label{sec:intro}
\input{intro.tex}

\section{Preliminaries}
\label{sec:prelim}
\input{prelim.tex}

\section{An Erd{\"o}s-P{\'o}sa Type Theorem}
\label{sec:erdosposa}
\input{erdos.tex}

\section{Fixed-Parameter Tractability and Kernelization Complexity}
\label{sec:fpt-ker1}
\input{fpt-ker1.tex}

\section{An Improved FPT Algorithm and A Smaller Kernel}
\label{sec:fpt-ker2}
\input{fpt-ker2.tex}

\section{Concluding Remarks}
\label{sec:concl}
\input{concl.tex}

\end{document}

%% file: abstract.tex
Given a directed graph $D$ on $n$ vertices and a positive integer $k$, the \cp\ problem is to determine whether $D$ has $k$ arc-disjoint cycles. This problem is known to be \W[1]-hard in general directed graphs. In this paper, we initiate a systematic study on the parameterized complexity of the problem restricted to tournaments. We show that the problem is fixed-parameter tractable and admits a polynomial kernel when parameterized by the solution size $k$. In particular, we show that it can be solved in $2^{\oo(k \log k)} n^{\oo(1)}$ time and has a kernel with $\oo(k)$ vertices. The primary ingredient in both these results is a min-max theorem that states that every tournament either contains $k$ arc-disjoint triangles or has a feedback arc set of size at most $6k$. Our belief is that this combinatorial result is of independent interest and could be useful in other problems related to cycles in tournaments.

%% file: intro.tex
Given a (directed or undirected) graph $G$ and a positive integer $k$, the \textsc{Disjoint Cycle Packing} problem is to determine whether $G$ has $k$ (vertex or arc/edge) disjoint cycles. Packing disjoint cycles is a fundamental problem in Graph Theory and Algorithm Design with applications in several areas. Since the publication of the classic Erd{\"o}s-P{\'o}sa theorem in 1965~\cite{ErdosPosa}, this problem has received significant scientific attention in various algorithmic realms. In particular, \vcp\ in undirected graphs is one of the first problems studied in the framework of parameterized complexity. In this framework, each problem instance is associated with a non-negative integer $k$ called {\em parameter}, and a problem is said to be {\em fixed-parameter tractable} (\FPT) if it can be solved in $f(k) n^{\oo(1)}$ time for some function $f$, where $n$ is the input size. For convenience, the running time $f(k)n^{\oo(1)}$ where $f$ grows super-polynomially with $k$ is denoted as $\ostar(f(k))$. A {\em kernelization algorithm} is a polynomial-time algorithm that transforms an arbitrary instance of the problem to an equivalent instance of the same problem whose size is bounded by some computable function $g$ of the parameter of the original instance. The resulting instance is called a {\em kernel} and if $g$ is a polynomial function, then it is called a {\em polynomial kernel} and we say that the problem admits a polynomial kernel. Kernelization typically involves applying a set of rules (called {\em reduction rules}) to the given instance to produce another instance. A reduction rule is said to be {\em safe} if it is sound and complete, i.e., applying it to the given instance produces an equivalent instance. In order to classify parameterized problems as being \FPT\ or not, the \W-hierarchy is defined: \FPT\ $\subseteq$ \W$[1] \subseteq$ \W$[2] \subseteq \dots \subseteq$ \XP. It is believed that the subset relations in this sequence are all strict, and a parameterized problem that is hard for some complexity class above \FPT\ in this hierarchy is said to be fixed-parameter intractable. Further details on parameterized algorithms can be found in~\cite{fpt-book,downey,fg}. 

\vcp\ in undirected graphs is \FPT\ with respect to the solution size $k$ \cite{OnDisCycBod,LMSZIcalp17} but has no polynomial kernel unless \NP\ $\subseteq$ co\NP/poly \cite{BTY11}. In contrast, \textsc{Edge-Disjoint Cycle Packing} in undirected graphs admits a kernel with $\oo(k \log k)$ vertices (and is therefore \FPT) \cite{BTY11}. On directed graphs, \vcp\ and \cp\ are equivalent and turn out to be \W[1]-hard \cite{vcp-red-ecp-red-vc, dag-edp}. Therefore, studying these problems on a subclass of directed graphs is a natural direction of research. Tournaments form a mathematically rich subclass of directed graphs with interesting structural and algorithmic properties \cite{gutin,moon}. A {\em tournament} is a directed graph in which there is a single arc between every pair of distinct vertices. Tournaments have several applications in modeling round-robin tournaments and in the study of voting theory and social choice theory. Further, the combinatorics of inclusion relations of tournaments is reasonably well-understood \cite{wqo-tour}. A seminal result in the theory of undirected graphs is the Graph Minor Theorem (also known as the Robertson and Seymour theorem) that states that undirected graphs are well-quasi-ordered under the {\em minor relation} \cite{gmt}. Developing a similar theory of inclusion relations of directed graphs has been a long-standing research challenge. However, there is such a result known for tournaments that states that tournaments are well-quasi-ordered under the {\em strong immersion relation} \cite{wqo-tour}. In fact, this result also holds for a superclass of tournaments, namely, semicomplete digraphs \cite{wqo-semi}. A {\em semicomplete digraph} is a directed graph in which there is at least one arc between every pair of distinct vertices. Many results (including some of the ones described in this work) for tournaments straightaway hold for semicomplete digraphs too. This is another reason why tournaments is one of the most well-studied classes of directed graphs. 

\textsc{Feedback Vertex Set} and \textsc{Feedback Arc Set} are two well-explored algorithmic problems on tournaments. A {\em feedback vertex (arc) set} is a set of vertices (arcs) whose deletion results in an acyclic graph. Given a directed graph and a positive integer $k$, \textsc{Feedback Arc (Vertex) Set} is the problem of determining if the graph has a set of at most $k$ arcs (vertices) whose deletion results in an acyclic graph. These problems are the dual problems of \vcp\ and \cp, respectively. They are known to be \NP-hard on tournaments \cite{fast-hard-alon,fast-hard,fast-hard3,speck} but are \FPT\ when parameterized by $k$ \cite{fastfast,karpinski,loc-ser,vr-ss}. Further, \textsc{Feedback Arc Set in Tournaments} has a kernel with $\oo(k)$ vertices \cite{jcss11} and \textsc{Feedback Vertex Set in Tournaments} has a kernel with $\oo(k^{1.5})$ vertices \cite{implicit-hs}. Though \textsc{Feedback Arc Set} and \textsc{Feedback Vertex Set} are intensively studied in tournaments, their duals have surprisingly not been considered in the literature until recently \cite{cpt-esa17,implicit-hs}. Any tournament that has a cycle also has a triangle \cite{digraphs}. Therefore, if a tournament has $k$ vertex-disjoint cycles, then it also has $k$ vertex-disjoint triangles. Thus, \vcp\ in tournaments is just packing vertex-disjoint triangles. A straightforward application of the {\em color coding} technique shows that this problem is \FPT\ and a kernel with $\oo(k^2)$ vertices is an immediate consequence of the quadratic element kernel known for \textsc{3-Set Packing} \cite{3-set-packing}. Recently, a kernel with $\oo(k^{1.5})$ vertices was shown for this problem using interesting variants and generalizations of the popular {\em expansion lemma} \cite{implicit-hs}. 

Focusing on \cp\ in tournaments, it is easy to verify that a tournament that has $k$ arc-disjoint cycles need not necessarily have $k$ arc-disjoint triangles. This observation hints that packing arc-disjoint cycles could be significantly harder than packing vertex-disjoint cycles. This is the starting point of our study. In this paper, we investigate the parameterized complexity of \cp\ in tournaments.

\defparprob
{\cpt ~(ACT)}
{A tournament $T$ and a positive integer $k$.}
{$k$}
{Do there exist $k$ arc-disjoint cycles in $T$?}

We show that \textsc{ACT} is \FPT\ and admits a polynomial kernel. En route, we discover an interesting min-max relation analogous to the classical Erd{\"o}s-P{\'o}sa theorem. In particular, we show the following results.
\begin{itemize}
\item A tournament $T$ has $k$ arc-disjoint cycles if and only if $T$ has $k$ arc-disjoint cycles each of length at most $2k+1$ (Lemma \ref{lem:short-cycle}).
\item Every tournament $T$ either contains $k$ arc-disjoint triangles or has a feedback arc set of size at most $6(k-1)$ (Theorem \ref{thm:lin-ep}).
\item \textsc{ACT} admits a kernel with $\oo(k)$ vertices (Theorem \ref{thm:cp-linear-kernel}).
 \item \textsc{ACT} can be solved in $\ostar(2^{\oo(k \log k)})$ time (Theorem \ref{thm:reduc-cp-edp}).
\end{itemize}
The paper is organized as follows. In Section \ref{sec:prelim}, we give some definitions related to directed graphs, cycles and tournaments. In Section \ref{sec:erdosposa}, we show the first two combinatorial results about tournaments. In Section \ref{sec:fpt-ker1}, we show that \textsc{ACT} is \FPT\ and admits a polynomial kernel. In Section \ref{sec:fpt-ker2}, we describe an improved \FPT\ algorithm and a linear vertex kernel for \textsc{ACT}. Finally, we conclude with some remarks in Section \ref{sec:concl}.

%% file: prelim.tex
The set $\{1,2,\dots,n\}$ is denoted by $[n]$. A {\em directed graph} (or {\em digraph}) is a pair consisting of a set $V$ of vertices and a set $A$ of arcs. An arc is specified as an ordered pair of vertices. We will consider only simple unweighted digraphs. For a digraph $D$, $V(D)$ and $A(D)$ denote the set of its vertices and the set of its arcs, respectively. Two vertices $u$, $v$ are said to be {\em adjacent}  in $D$ if $(u,v) \in A(D)$ or $(v,u) \in A(D)$. For an arc $e=(u,v)$, $\head(e)$ denotes $v$ and $\tail(e)$ denotes $u$. For a vertex $v \in V(D)$, its {\em out-neighborhood}, denoted by $N^{+}(v)$, is the set $\{u \in V(D) \mid (v,u) \in A(D)\}$ and its {\em in-neighborhood}, denoted by $N^{-}(v)$, is the set $\{u \in V(D) \mid (u,v) \in A(D)\}$. For a set of arcs $F$, $V(F)$ denotes the union of the sets of endpoints of arcs in $F$. For a set $X \subseteq V(D) \cup E(D)$, $D-X$ denotes the digraph obtained from $D$ by deleting $X$. 

A {\em path} $P$ in $D$ is a sequence $(v_1,\dots,v_k)$ of distinct vertices such that for each $i \in [k-1]$, $(v_i,v_{i+1}) \in A(D)$. The set $\{v_1,\dots,v_k\}$ is denoted by $V(P)$ and the set $\{(v_i,v_{i+1}) \mid i \in [k-1]\}$ is denoted by $A(P)$. A path $P$ is called an {\em induced} (or {\em chordless}) path if there is no arc in $D$ that is between two non-consecutive vertices of $P$. A {\em cycle} $C$ in $D$ is a sequence $(v_1,\dots,v_k)$ of distinct vertices such that $(v_1,\dots,v_k)$ is a path and $(v_k,v_1) \in A(D)$. The set $\{v_1,\dots,v_k\}$ is denoted by $V(C)$ and the set $\{(v_i,v_{i+1}) \mid i \in [k-1]\} \cup \{(v_k,v_1)\}$ is denoted by $A(C)$. A cycle $C=(v_1,\dots,v_k)$ is called an {\em induced} (or {\em chordless}) cycle if there is no arc in $D$ that is between two non-consecutive vertices of $C$ with the exception of the arc $(v_k,v_1)$. The length of a path or cycle $X$ is the number of vertices in it and is denoted by $|X|$. For a set $\mc$ of paths or cycles, $A(\mc)$ denotes the set $\{e \in A(D) \mid \exists C \in \mc, e \in A(C)\}$. A cycle on three vertices is called a {\em triangle}. A digraph is said to be {\em triangle-free} if it has no triangles.

A digraph is called a {\em directed acyclic graph} if it has no cycles. Any directed acyclic graph $D$ has an ordering $\sigma$ called {\em topological ordering} of its vertices such that for each $(u,v) \in A(D)$, $\sigma(u)<\sigma(v)$ holds. A {\em feedback vertex (arc) set} is a set of vertices (arcs) whose deletion results in an acyclic graph. For a digraph $D$, let $\fas(D)$ denote the size of a minimum feedback arc set of $D$. A {\em tournament} $T$ is a digraph in which for every pair $u,v$ of distinct vertices either $(u,v) \in A(T)$ or $(v,u) \in A(T)$ but not both. A tournament is called {\em transitive} if it is a directed acyclic graph and a transitive tournament has a unique topological ordering.

%% file: erdos.tex
The classical Erd{\"o}s-P{\'o}sa theorem for cycles in undirected graphs states that there exists a function $f(k)=\oo(k\log k)$ such that for each non-negative integer $k$, every undirected graph either contains $k$ vertex-disjoint cycles or has a feedback vertex set consisting of $f(k)$ vertices \cite{ErdosPosa}. An interesting consequence of this theorem is that it leads to an \FPT\ algorithm for \vcp. It is well known that the treewidth ($\mathrm{\it tw}$) of a graph is not larger than the size of its feedback vertex set, and that a naive dynamic programming scheme solves \vcp\ in $2^{\oo(\mathrm{\it tw}\log \mathrm{\it tw})} n$ time and exponential space (see, e.g., \cite{fpt-book}). Thus, the existence of a $2^{\oo(k\log^2k)} n$ time algorithm that uses exponential space can be viewed as a direct consequence of the Erd{\"o}s-P{\'o}sa theorem (see \cite{LMSZIcalp17} for more details).  

In this section, we show that that there exists a function $h(k)$ such that for each non-negative integer $k$, every tournament either contains $k$ arc-disjoint cycles or has a feedback arc set consisting of $h(k)$ arcs. First, we show that $h(k)$ is $\oo(k^2)$ and then improve it to $\oo(k)$. The following result is crucial in proving the former. 

\begin{lemma}
\label{lem:short-cycle}
Let $k$ and $r$ be positive integers such that $r \leq k$. If a tournament $T$ contains a set $\mc$ of $r$ arc-disjoint cycles, then it also contains a set $\mc^*$ of $r$ arc-disjoint cycles each of length at most $2k+1$.
\end{lemma}

\begin{proof}
Let $\mc$ be a set of $r$ arc-disjoint cycles in $T$ that minimizes $\sum_{C \in \mc} |C|$. If every cycle in $\mc$ is a triangle, then the claim trivially holds. Otherwise, let $C$ be a longest cycle in $\mc$ and let $\ell$ denote its length. Let $v_i,v_j$ be a pair of non-consecutive vertices in $C$. Then, either $(v_i,v_j) \in A(T)$ or $(v_j,v_i) \in A(T)$. In any case, the arc $e$ between $v_i$ and $v_j$ along with $A(C)$ forms a cycle $C'$ of length less than $\ell$ with $A(C') \setminus \{e\} \subset A(C)$.  By our choice of $\mc$, this implies that $e$ is an arc in some other cycle $\widehat{C} \in \mc$. This property is true for the arc between any pair of non-consecutive vertices in $C$. Therefore, we have ${\ell \choose 2}-\ell \leq \ell(k-1)$ leading to $\ell \leq 2k+1$.
\end{proof}

This lemma essentially shows that it suffices to determine the existence of $k$ arc-disjoint cycles in $T$ each of length at most $2k+1$ in order to determine if $(T,k)$ is a yes-instance of \textsc{ACT}. This leads to the following quadratic Erd{\"o}s-P{\'o}sa bound. Recall that for a digraph $D$, $\fas(D)$ denotes the size of a minimum feedback arc set of $D$. 

\begin{theorem}
\label{thm:quad-ep}
For every non-negative integer $k$, every tournament $T$ either contains $k$ arc-disjoint cycles or has a feedback arc set of size $\oo(k^2)$.
\end{theorem}
\begin{proof}
Suppose $\mc$ is a maximal set of arc-disjoint cycles in $T$. If $|\mc| \geq k$, then the claim holds. Otherwise, from Lemma \ref{lem:short-cycle}, we may assume that each cycle in $\mc$ is of length at most $2k+1$. Let $D$ denote the digraph obtained from $T$ by deleting the arcs that are in some cycle in $\mc$. Clearly, $D$ is acyclic as $\mc$ is maximal. Then, it follows that $\fas(T) \leq (2k+1)(k-1)$. 
\end{proof}

Next, we strengthen this result to arrive at a linear min-max bound. We will use the following lemma in the process. For a digraph $D$, let $\Lambda(D)$ denote the number of non-adjacent pairs of vertices in $D$. That is, $\Lambda(D)$ is the number of pairs $u,v$ of vertices of $D$ such that neither $(u,v) \in A(D)$ nor $(v,u) \in A(D)$. 

 
\begin{lemma}
\label{lem:erdos-posa}
Let $D$ be a triangle-free digraph in which for every pair $u,v$ of distinct vertices, at most one of $(u,v)$ or $(v,u)$ is in $A(D)$. Then, $\fas(D) \leq \Lambda(D)$.
\end{lemma}

\begin{proof}
We will prove the claim by induction on $|V(D)|$. The claim trivially holds for $|V(D)| \leq 2$. Suppose $|V(D)| \geq 3$. First, we apply a simple preprocessing rule on $D$. If $D$ has a vertex $v$ that either has no in-neighbours or has no out-neighbours, then we delete $v$ from $D$ to get the digraph $D'$. Clearly, there is no cycle in $D$ that contains $v$ and thus $\fas(D)=\fas(D')$. Therefore, subsequently, we may assume that for every vertex $v \in V(D)$, $N^{+}(v)\neq \emptyset$ and $N^{-}(v)\neq \emptyset$.

For a vertex $v \in V(D)$, we define $\first(v)$ to be the number of induced paths of length 3 with $v$ as the first vertex. Similarly, we define $\second(v)$ to be the number of induced paths of length 3 with $v$ as the second vertex. We claim that $\sum_{v \in V(D)}\first(v)=\sum_{v \in V(D)} \second(v)$. Consider an induced path $P=(u,v,w)$ of length 3 in $D$. Then, $P$ contributes 1 to $\first(u)$ and does not contribute to $\first(x)$ for any $x \neq u$. Further, $P$ contributes 1 to $\second(v)$ and does not contribute to $\second(x)$ for any $x \neq v$. Therefore, $P$ contributes 1 to $\sum_{v \in V(D)}\first(v)$ and $\sum_{v \in V(D)} \second(v)$. Hence, $\sum_{v \in V(D)}\first(v)=\sum_{v \in V(D)} \second(v)$. It now follows that there is a vertex $u \in V(D)$ such that $\first(u) \leq \second(u)$. 

Define the sets $I_u$, $O_u$ and $R_u$ as $I_u=N^{-}(v)$, $O_u=N^{+}(v)$ and $R_u=V(D)\setminus (I_u \cup O_u)$. That is, $I_u$ is the set of in-neighbours of $u$, $O_u$ is the set of out-neighbours of $u$ and $R_u$ is the set of vertices that are not adjacent with $u$. Observe that $I_u\neq \emptyset$ and $O_u\neq \emptyset$. Let $D_1$ and $D_2$ be the subgraphs $D[I_u \cup R_u]$ and $D[O_u]$, respectively. Then, as $D_1$ and $D_2$ are vertex-disjoint induced subgraphs of $D$, we have $\Lambda(D)\geq \Lambda(D_1)+\Lambda(D_2)$. Now, any induced path $P=(x,u,y)$ of length 3 in $D$ with $u$ as the second vertex satisfies the property that $x \in V(D_1)$ and $y \in V(D_2)$. Further, $(x,y),(y,x) \notin A(D)$ due to the facts that $P$ is an induced path and $D$ is triangle-free. Then, $\second(u)$ is the number of pairs $x,y$ of non-adjacent vertices with $x \in I_u,y \in O_u$. Therefore, $\Lambda(D)\geq \Lambda(D_1)+\Lambda(D_2)+ \second(u)$ as $I_u \subseteq V(D_1)$, $O_u \subseteq V(D_2)$ and $V(D_1) \cap V(D_2)=\emptyset$. 

Let $E$ denote the set of arcs $(x,y)$ in $D$ with $x \in O_u$ and $y \in R_u$. Let $F_1$ and $F_2$ be feedback arc sets of $D_1$ and $D_2$, respectively. We  claim that $F=F_1 \cup F_2 \cup E$ is a feedback arc set of $D$. If there is a cycle $C$ in the graph obtained from $D$ by removing arcs in $F$, then $C$ has an arc $(p,q)$ with $p \in V(D_1), q \in V(D_2)$ and an arc $(r,s)$ with $r \in V(D_2), s \in V(D_1)$. However, as $D$ has no triangle, any arc $(x,y)$ with $x \in V(D_2)$ and $y \in V(D_1)$ satisfies $x \in O_u$ and $y \in R_u$. That is, $(r,s) \in E$ leading to a contradiction. Therefore, it follows that $\fas(D) \leq \fas(D_1)+\fas(D_2)+|E|$. Note that $|E| = \first(u)$ as any induced path $P=(u,v,w)$ starting at $u$ satisfies $v \in O_u$, $w \in R_u$ and any arc $(x,y)$ with $x \in O_u, y\in R_u$ corresponds to an induced path $(u,x,y)$ starting at $u$. Also, by the choice of $u$, we have $\first(u) \leq \second(u)$. Therefore, $\fas(D) \leq \fas(D_1)+\fas(D_2)+\second(u)$. By induction hypothesis, we have $\fas(D_1) \leq \Lambda(D_1)$ and $\fas(D_2) \leq \Lambda(D_2)$. Hence, $\fas(D) \leq \Lambda(D_1)+\Lambda(D_2)+\second(u)$. As $\Lambda(D)\geq \Lambda(D_1)+\Lambda(D_2)+ \second(u)$, we have $\fas(D) \leq \Lambda(D)$.
\end{proof}

This leads to the following main result of this section.
\begin{theorem}
\label{thm:lin-ep}
For every non-negative integer $k$, every tournament $T$ either contains $k$ arc-disjoint triangles or has a feedback arc set of size at most $6(k-1)$ that can be obtained in polynomial time.
\end{theorem}
\begin{proof}
Suppose $\mc$ is a maximal set of arc-disjoint triangles in $T$ with $|\mc| \leq k-1$. Let $D$ denote the digraph obtained from $T$ by deleting the arcs that are in some triangle in $\mc$. Clearly, $D$ has no triangle and $\Lambda(D) \leq 3(k-1)$. From Lemma \ref{lem:erdos-posa}, we have $\fas(D) \leq 3(k-1)$. Also, if $F$ is a feedback arc set of $D$, then $F \cup A(\mc)$ is a feedback arc set of $T$. Therefore, $\fas(T) \leq 6(k-1)$.
\end{proof}

We will use this result crucially in showing that \textsc{ACT} can be solved in $\ostar(2^{\oo(k \log k)})$ time and admits a kernel with $\oo(k)$ vertices.

%% file: fpt-ker1.tex
In this section, we show that \textsc{ACT} is \FPT\ and admits a polynomial kernel. We show that the first result is a direct consequence of Lemma \ref{lem:short-cycle} and the second follows from Theorem \ref{thm:lin-ep}.
\subsection{An FPT Algorithm}
Consider an instance $\mi=(T,k)$ of \textsc{ACT}. Let $n$ denote $|V(T)|$ and $m$ denote $|A(T)|$. Suppose $\mi$ is a yes-instance and $\mc$ is a set of $k$ arc-disjoint cycles in $T$. From Lemma \ref{lem:short-cycle},  we may assume that the total number of arcs that are in cycles in $\mc$ is at most $(2k+1)k$. Using this observation, we proceed as follows. We color the arcs of $T$ uniformly at random from the color set $[\ell]$ where $\ell=2k^2+k$. Let $\chi:A(T) \rightarrow [\ell]$ denote this coloring. 

\begin{proposition}[\cite{color-coding}]
\label{prop:colorful}
If $E$ is a subset of $A(T)$ of size $\ell$, then the probability that the arcs in $E$ are colored with pairwise distinct colors is at least $e^{-\ell}$.
\end{proposition}

Next, we define the notion of a colorful solution for our problem.

\begin{definition}{\bf (Colorful set of cycles)} 
\label{def:colorful-sol}
A set $\mc$ of arc-disjoint cycles in $T$ that satisfies the property that for any two (not necessarily distinct) cycles $C,C' \in \mc$ and for any two distinct arcs $e \in A(C), e' \in A(C')$, $\chi(e) \neq \chi(e')$ holds is said to be a {\em colorful set of cycles}.
\end{definition}

Rephrasing Proposition \ref{prop:colorful} in the context of our problem, we have the following observation.

\begin{observation}
\label{obs:colorful-opt}
If $\mc$ is a solution of $\mi$ with the property that for each $C \in \mc$, $|C| \leq 2k+1$, then $\mc$ is a colorful set of cycles in $T$ with probability at least $e^{-\ell}$.
\end{observation}

Armed with the guarantee that a solution (if one exists) of $\mi$ is colorful with sufficiently high probability, we focus on finding a colorful set of cycles in $T$. 

\begin{lemma}
\label{lem:colorful-cycles}
If $T$ has a colorful set of $k$ cycles, then such a set can be obtained in $\ell! n^{\oo(1)}$ time.
\end{lemma}
\begin{proof}
Consider a permutation $\sigma$ of $[\ell]$. For each $i \in [k]$, let $D^\sigma_i$ denote the subgraph of $T$ with $V(D^\sigma_i)=V(T)$ and $A(D^\sigma_i)=A(T) \cap \{ (u,v) \in A(T) \mid 2(i-1)k+i \leq \sigma(\chi((u,v))) \leq 2k+2(i-1)k+i\}$. That is, $A(D^\sigma_1)$ is the set of arcs of $T$ that are colored with the first $2k+1$ colors, $A(D^\sigma_2)$ is the set of arcs of $T$ that are colored with the next $2k+1$ colors and so on. For each $i \in [k]$, let $C^\sigma_i$ denote a cycle (if one exists) in $D^\sigma_i$. Let $\mc_\sigma$ denote the set $\{C^\sigma_i \mid i \in [k]\}$. For each permutation $\sigma$ of $[\ell]$, we compute the corresponding set $\mc_\sigma$. If $T$ has a colorful set of $k$ cycles, then $|\mc_\pi|=k$ for some permutation $\pi$ of $[\ell]$. Therefore, by computing $\mc_\pi$ for every permutation $\pi$ of $[\ell]$, we can obtain a colorful set of $k$ cycles in $T$ (if one exists). 
\end{proof}

Using the standard technique of derandomization of color coding based algorithms \cite{color-coding,fpt-book,splitters}, we have the following result by taking $m=|A(T)|$.

\begin{proposition}[\cite{color-coding,fpt-book,splitters}]
\label{prop:perfect-family}
Given integers $m,\ell \geq 1$, there is a family $\mf_{m,\ell}$ of coloring functions $\chi:A(T) \rightarrow [\ell]$ of size $e^{\ell} {\ell}^{\oo(\log \ell)} \log m$ that can be constructed in $e^{\ell} {\ell}^{\oo(\log \ell)} m \log m$ time satisfying the following property: for every set $E \subseteq A(T)$ of size $\ell$, there is a function $\chi \in \mf_{m,\ell}$ such that $\chi(e) \neq \chi(e')$ for any two distinct arcs $e,e' \in E$. 
\end{proposition}

Then, we have the following result.

\begin{theorem}
\label{thm:col-code-algo}
\textsc{ACT} can be solved in $\ostar(2^{\oo(k^2 \log k)})$ time.
\end{theorem}
\begin{proof}
Consider an instance $\mi=(T,k)$ of \textsc{ACT}. Let $\ell=2k^2+k$. First, we compute the family $\mf_{m,\ell}$ of $e^{\ell} {\ell}^{\oo(\log \ell)} \log m$ coloring functions using Proposition \ref{prop:perfect-family} where $m$ is the number of arcs in $T$. Then, for each coloring function $\chi:A(T) \rightarrow [\ell]$ in $\mf_{m,\ell}$, we determine if $T$ has a colorful set of $k$ cycles using Lemma \ref{lem:colorful-cycles}. Due to the properties of $\mf_{m,\ell}$ guaranteed by Proposition \ref{prop:perfect-family}, it follows that $\mi$ is a yes-instance if and only if $T$ has a set of $k$ cycles that is colorful with respect to at least one of the coloring functions. The overall running time is $\ostar(2^{\oo(k^2 \log k)})$.
\end{proof}

Observe that the running time of the algorithm to find a colorful set of $k$ cycles can be improved to $2^\ell n^{\oo(1)}$ by employing a standard dynamic programming scheme. This will result in an $\ostar(2^{\oo(k^2)})$ time algorithm for \textsc{ACT}. However, we skip the details of the same as we will describe an $\ostar(2^{\oo(k \log k)})$ time algorithm for \textsc{ACT} in Section \ref{sec:fpt-ker2}.

%
%

\subsection{A Polynomial Kernel}
Now, we show that \textsc{ACT} admits a polynomial kernel. We use Theorem \ref{thm:lin-ep} to describe a quadratic vertex kernel. 

\begin{theorem}
\label{thm:cp-quad-kernel}
\textsc{ACT} admits a kernel with $\oo(k^2)$ vertices.  
\end{theorem}
\begin{proof}
Let $(T,k)$ denote an instance of \textsc{ACT}. From Theorem \ref{thm:lin-ep}, we know that $T$ has either $k$ arc-disjoint triangles or a feedback arc set $F$ of size at most $6(k-1)$. In the former case, we return a trivial yes-instance of constant size as the kernel. In the latter case, $S=V(F)$ is a feedback vertex set of $T$ of size at most $12k$. Let $D$ denote the transitive tournament $T-S$ and $\delta$ denote its unique topological ordering. Observe that for each $v \in S$, the subtournament of $T$ induced by $V(D) \cup \{v\}$ is also transitive. If there is a cycle in $D \cup \{v\}$, then this cycle (which is also a cycle in $T$) has no arc from $F$ leading to a contradiction. 

For each $v \in S$, let $R(v)$ be the set of first (with respect to $\delta$) $2k+1$ vertices in $N^{+}(v)$. Let $T'$ be the subtournament of $T$ induced by $S \cup \{R(v) \mid v \in S\}$. Clearly, $T'$ has $\oo(k^2)$ vertices. We claim that $(T',k)$ is the required kernel of $(T,k)$. We need to show that $T$ has $k$ arc-disjoint cycles if and only if $T'$ has $k$ arc-disjoint cycles. The reverse direction of the claim holds trivially. Let us now prove the forward direction. Suppose $T$ has a set of $k$ arc-disjoint cycles. Among all such sets, let $\mc$ be one that minimizes $\sum_{C \in \mc} |V(C) \cap (V(T) \setminus V(T'))|$. Suppose there is a cycle $C$ in $\mc$ that is not in $T'$. Then, there is a vertex $v_i \in V(C)$ that is not in $T'$. As argued earlier, any cycle in $T$ has at least two vertices from $S$. Let $x$ and $y$ be two such vertices in $C$ where $(x,v_1,\dots,v_i,\dots,v_q,y)$ is a path in $C$ from $x$ to $y$ with internal vertices from $V(D)$. 

The subtournaments $\widehat{T}=D \cup \{x\}$ and $\widetilde{T}=D \cup \{y\}$ are transitive with unique topological orderings $\sigma$ and $\pi$, respectively. Observe that for all distinct $u,v \in V(D)$, $\pi(u)<\pi(v)$ if and only if $\sigma(u)<\sigma(v)$. As $(x,v_1,\dots,v_i,\dots,v_q)$ is a path in $\widehat{T}$, it follows that $\sigma(x)<\sigma(v_j)$ for each $j \in [q]$. Similarly, as $(v_1,\dots,v_i,\dots,v_q,y)$ is a path in $\widetilde{T}$, we have $\pi(y)>\pi(v_j)$ for each $j \in [q]$. As $v_i \notin V(T')$, it follows that $v_i \notin R(x)$ and $|R(x)| = 2k+1$. Then, there is at least one vertex $z$ in $R(x)$ such that the arcs $(x,z)$ and $(z,y)$ are not in any cycle in $\mc$. Now, $\sigma(z)<\sigma(v_i)$ as $z,v_i \in N^{+}(x)$, $v_i \notin R(x)$ and $z \in R(x)$. Thus, we have $\pi(z)<\pi(v_i)$. As $\pi(v_i)<\pi(y)$, it follows that $(z,y) \in A(T)$ as $\pi(z)<\pi(y)$. Then, by replacing the path $(x,v_1,\dots,v_i,\dots,v_q,y)$ by $(x,z,y)$, we obtain another set $\mc'$ of $k$ arc-disjoint cycles such that $\sum_{C \in \mc} |V(C) \cap (V(T) \setminus V(T'))|>\sum_{C \in \mc'} |V(C) \cap (V(T) \setminus V(T'))|$. However, this leads to a contradiction by the choice of $\mc$.
\end{proof}

%% file: fpt-ker2.tex
Next, we show that \textsc{ACT} can be solved in $\ostar(2^{\oo(k \log k)})$ time and admits a kernel with $\oo(k)$ vertices. 


\subsection{A Linear Vertex Kernel}
We show that the linear kernelization described in \cite{jcss11} for \textsc{Feedback Arc Set in Tournaments} also leads to a linear kernelization for our problem. In order to describe the kernel, we need to state some terminology defined in \cite{jcss11}. Let $T$ be a tournament on $n$ vertices. First, we apply the following reduction rule.

\begin{reduction rule}
\label{rule1}
If a vertex $v$ is not in any cycle, then delete $v$ from $T$.
\end{reduction rule}

This rule is clearly safe as our goal is to find $k$ cycles and $v$ cannot be in any of them. To describe our next rule, we need to state some terminology and a lemma known from \cite{jcss11}. For an ordering $\sigma$ of $V(T)$, let $T_\sigma$ denote the tournament $T$ whose vertices are ordered according to $\sigma$. Clearly, $V(T_\sigma)=V(T)$ and $A(T_\sigma)=A(T)$ since $T$ and $T_\sigma$ denote the same tournament. An arc $(u,v) \in A(T_\sigma)$ is called a {\em back arc} if $\sigma(u)>\sigma(v)$ and it is called a {\em forward arc} otherwise. An {\em interval} is a consecutive set of vertices in $T_\sigma$. 

\begin{lemma}[\cite{jcss11}]\footnote{Lemma \ref{lem:safe-part} is Lemma 3.9 of \cite{jcss11} that has been rephrased to avoid the use of several definitions and terminology introduced in \cite{jcss11}.}
\label{lem:safe-part}
Let $T_\sigma$ be an ordered tournament on which Reduction Rule \ref{rule1} is not applicable. Let $B$ denote the set of back arcs in $T_\sigma$ and $E$ denote the set of arcs in $T_\sigma$ with endpoints in different intervals. If $|V(T_\sigma)| \geq 2 |B|+1$, then there exists a partition $\mj$ of $V(T_\sigma)$ into intervals with the following properties that can be computed in polynomial time. 
\begin{itemize}
\item There is at least one arc $e=(u,v) \in A(T)$ with $e \in B \cap E$.
\item There are $|B \cap E|$ arc-disjoint cycles using only arcs in $E$. 
\end{itemize}
\end{lemma}

Our reduction rule that is based on this lemma is as follows.

\begin{reduction rule}
\label{rule2}
Let $T_\sigma$ be an ordered tournament on which Reduction Rule \ref{rule1} is not applicable. Let $B$ denote the set of back arcs in $T_\sigma$ and $E$ denote the set of arcs in $T_\sigma$ with endpoints in different intervals. Let $\mj$ be a partition of $V(T_\sigma)$ into intervals satisfying the properties specified in Lemma \ref{lem:safe-part}. Reverse all arcs in $B \cap E$ and decrease $k$ by $|B \cap E|$. 
\end{reduction rule}

\begin{lemma}
Reduction Rule \ref{rule2} is safe.
\end{lemma}
\begin{proof}
Let $T'_\sigma$ be the tournament obtained from $T_\sigma$ by reversing all arcs in $B \cap E$. Suppose $T'_\sigma$ has $k-|B \cap E|$ arc-disjoint cycles. Then, it is guaranteed that each such cycle is completely contained in an interval. This is due to the fact that $T'_\sigma$ has no back arc with endpoints in different intervals. Indeed, if a cycle in $T'_\sigma$ uses a forward (back) arc with endpoints in different intervals, then it also uses a back (forward) arc with endpoints in different intervals. It follows that for each arc $(u,v) \in E$, neither $(u,v)$ nor $(v,u)$ is used in these $k-|B \cap E|$ cycles. Hence, these $k-|B \cap E|$ cycles in $T'_\sigma$ are also cycles in $T_\sigma$. Then, we can add a set of $|B \cap E|$ cycles obtained from the second property of Lemma \ref{lem:safe-part} to these $k-|B \cap E|$ cycles to get $k$ cycles in $T_\sigma$. Conversely, consider a set of $k$ cycles in $T_\sigma$. As argued earlier, we know that the number of cycles that have an arc that is in $E$ is at most $|B \cap E|$. The remaining cycles (at least $k-|B \cap E|$ of them) do not contain any arc that is in $E$, in particular, they do not contain any arc from $B \cap E$. Therefore, these cycles are also cycles in $T'_\sigma$.
\end{proof}

\begin{theorem}
\label{thm:cp-linear-kernel}
\textsc{ACT} admits a kernel with $\oo(k)$ vertices. 
\end{theorem}
\begin{proof}
Let $(T,k)$ denote the instance obtained from the input instance by applying Reduction Rule \ref{rule1} exhaustively. From Lemma \ref{thm:lin-ep}, we know that either $T$ has $k$ arc-disjoint triangles or has a feedback arc set of size at most $6(k-1)$ that can be obtained in polynomial time. In the first case, we return a trivial yes-instance of constant size as the kernel. In the second case, let $F$ be the feedback arc set of size at most $6(k-1)$ of $T$. Let $\sigma$ denote a topological ordering of the vertices of the directed acyclic graph $T-F$. As $V(T-F)=V(T)$, $\sigma$ is an ordering of $V(T)$ such that $T_\sigma$ has at most $6(k-1)$ back arcs. If $|V(T_\sigma)| \geq 12k-11$, then from Lemma \ref{lem:safe-part}, there is a partition of $V(T_\sigma)$ into intervals with the specified properties. Therefore, Reduction Rule \ref{rule2} is applicable (and the parameter drops by at least 1). When we obtain an instance where neither of the Reduction Rules \ref{rule1} and \ref{rule2} is applicable, it follows that the tournament in that instance has at most $12k$ vertices. 
\end{proof}

\subsection{A Faster FPT Algorithm}
Here, we show that \textsc{ACT} can be solved in $\ostar(2^{\oo(k \log k)})$ time. The idea is to reduce the problem to the following \textsc{Arc-Disjoint Paths} problem in directed acyclic graphs.

\defparprob{Arc-Disjoint Paths}
{A digraph $D$ on $n$ vertices and $k$ ordered pairs $(s_1,t_1),\dots,(s_k,t_k)$ of vertices of $D$.}
{$k$}
{Do there exist arc-disjoint paths $P_1,\dots,P_k$ in $D$ such that $P_i$ is a path from $s_i$ to $t_i$ for each $i \in [k]$?}

On directed acyclic graphs, \textsc{Arc-Disjoint Paths} is known to be \NP-complete \cite{dag-edp-npc}, \W[1]-hard \cite{dag-edp} and solvable in $n^{\oo(k)}$ time \cite{dag-edp-xp}. Despite its fixed-parameter intractability, we will show that we can use the $n^{\oo(k)}$ algorithm to describe another (and faster) \FPT\ algorithm for \textsc{ACT}.

\begin{theorem}
\label{thm:reduc-cp-edp}
\textsc{ACT} can be solved in $\ostar(2^{\oo(k \log k)})$ time.
\end{theorem}
\begin{proof}
Consider an instance $(T,k)$ of \textsc{ACT}. Using Theorem \ref{thm:cp-linear-kernel}, we obtain a kernel $\mi=(\widehat{T},\widehat{k})$ such that $\widehat{T}$ has $\oo(k)$ vertices. Further, $\widehat{k} \leq k$. By definition, $(T,k)$ is a yes-instance if and only if $(\widehat{T},\widehat{k})$ is a yes-instance. Using Theorem \ref{thm:lin-ep}, we know that $\widehat{T}$ either contains $\widehat{k}$ arc-disjoint triangles or has a feedback arc set of size at most $6(\widehat{k}-1)$ that can be obtained in polynomial time. If Theorem \ref{thm:lin-ep} returns a set of $\widehat{k}$ arc-disjoint triangles in $\widehat{T}$, then we declare that $(T,k)$ is a yes-instance. 

Otherwise, let $\widehat{F}$ be the feedback arc set of size at most $6(\widehat{k}-1)$ returned by Theorem \ref{thm:lin-ep}. Let $D$ denote the (acyclic) digraph obtained from $\widehat{T}$ by deleting $\widehat{F}$. Observe that $D$ has $\oo(k)$ vertices. Suppose $\widehat{T}$ has a set $\mc=\{C_1,\dots,C_{\widehat{k}}\}$ of $\widehat{k}$ arc-disjoint cycles. For each $C \in \mc$, we know that $A(C) \cap \widehat{F} \neq \emptyset$ as $\widehat{F}$ is a feedback arc set of $\widehat{T}$. We can guess that subset $F$ of $\widehat{F}$ such that $F=\widehat{F} \cap A(\mc)$. Then, for each cycle $C_i \in \mc$, we can guess the arcs $F_i$ from $F$ that it contains and also the order $\sigma_i$ in which they appear. This information is captured as a partition $\mf$ of $F$ into $\widehat{k}$ sets, $F_1$ to $F_{\widehat{k}}$ and the set $\{\sigma_1,\dots,\sigma_{\widehat{k}}\}$ of permutations where $\sigma_i$ is a permutation of $F_i$ for each $i \in [\widehat{k}]$. Any cycle $C_i$ that has $F_i \subseteq F$ contains a $(v,x)$-path between every pair $(u,v)$, $(x,y)$ of consecutive arcs of $F_i$ with arcs from $A(D)$. That is, there is a path from $\head(\sigma_i^{-1}(j))$ and $\tail(\sigma_i^{-1}((j+1) \mod |F_i|))$ with arcs from $D$ for each $j \in [|F_i|]$. The total number of such paths in these $\widehat{k}$ cycles is $\oo(|F|)$ and the arcs of these paths are contained in $D$ which is a (simple) directed acyclic graph. 

The number of choices for $F$ is $2^{|\widehat{F}|}$ and the number of choices for a partition $\mf=\{F_1,\dots,F_{\widehat{k}}\}$ of $F$ and a set $X=\{\sigma_1,\dots,\sigma_{\widehat{k}}\}$ of permutations is $2^{\oo(|\widehat{F}| \log |\widehat{F}|)}$. Once such a choice is made, the problem of finding $\widehat{k}$ arc-disjoint cycles in $\widehat{T}$ reduces to the problem of finding $\widehat{k}$ arc-disjoint cycles $\mc=\{C_1,\dots,C_{\widehat{k}}\}$ in $\widehat{T}$ such that for each $1 \leq i \leq \widehat{k}$ and for each $1 \leq j \leq |F_i|$, $C_i$ has a path $P_{ij}$ between $\head(\sigma_i^{-1}(j))$ and $\tail(\sigma_i^{-1}((j+1) \mod |F_i|))$ with arcs from $D=\widehat{T}-\widehat{F}$. This problem is essentially finding $r=\oo(|\widehat{F}|)$ arc-disjoint paths in $D$ and can be solved in ${|V(D)|}^{\oo(r)}$ time using the algorithm in \cite{dag-edp-xp}. Therefore, the overall running time of the algorithm is $\ostar(2^{\oo(k \log k)})$ as $|V(D)|=\oo(k)$ and $r=\oo(k)$.
\end{proof}

%% file: concl.tex
We initiated the parameterized complexity study of the \cp\ problem on tournaments. We showed that it is \FPT\ when parameterized by the solution size and admits a linear vertex kernel. However, the classical complexity status of the problem is still open, i.e, we do not know if it is \NP-hard or not. Resolving the same is a natural future research direction. We conjecture that it is indeed \NP-hard. Note that the classical complexity status of the dual problem (\textsc{Feedback Arc Set in Tournaments}) was a long-standing open problem until it was shown to be \NP-hard \cite{fast-hard-alon,fast-hard,fast-hard3}. 
